\documentclass[12pt,oneside,reqno]{amsart}
\usepackage[final]{graphicx}
\usepackage{amsfonts}
\usepackage{pdfsync}
\usepackage{amsmath}
\usepackage{amssymb}
\usepackage{hyperref}
\usepackage{bm}
\usepackage{caption}
\usepackage{enumerate} 
\usepackage{dcolumn} 
\usepackage{color}
\usepackage{physics}
\usepackage[dvipsnames]{xcolor}

\newtheorem{theorem}{Theorem}[section]
\newtheorem*{theorem*}{Theorem}

\newtheorem{corollary}[theorem]{Corollary}
\newtheorem{proposition}[theorem]{Proposition}

\usepackage{pdfpages}

  \newcommand{\R}{\mathbb{R}}
  \newcommand{\C}{\mathbb{C}}
\newcommand{\Dom}{\operatorname{Dom}}

\begin{document}
\title[A block--diagonal form for four--component operators describing GQDs]{\bf A block--diagonal form for four--component operators describing graphene quantum dots}

\author[Benguria]{Rafael~D.~Benguria$^1$}

\author[Stockmeyer]{Edgardo~Stockmeyer$^2$}

\author[Vallejos]{Crist\'obal~Vallejos$^3$}

\author[Van Den Bosch]{Hanne Van Den Bosch$^4$}

\address{$^1$ Instituto  de F\'\i sica, Pontificia Universidad Cat\' olica de Chile,}
\email{{rbenguri@uc.cl}}

\address{$^2$ Instituto  de F\'\i sica, Pontificia Universidad Cat\' olica de Chile,}
\email{{stock@fis.puc.cl}}

\address{$^3$ Physics Department, Pennsylvania State University,}
\email{{civ5104@psu.edu}}

\address{$^4$ Departamento de Ingenier\'\i a Matemática, Universidad de Chile, and Centro de Modelamiento Matemático, Universidad de Chile \& IRL-CNRS 2807.}
\email{{hvdbosch@dim.uchile.cl}}

\begin{abstract} 
We consider four--component Dirac operators on domains in the plane. With suitable boundary conditions, these operators describe graphene quantum dots. The most general boundary conditions  are defined by a matrix depending on four real parameters. For operators with constant boundary parameters we show that the Hamiltonian is unitary equivalent to two copies of the two--component operator. This allows to extend the known results for this type of operators to the four--component case. 
As an application, we identify the boundary conditions from the tight--binding model for graphene that give rise to a block--diagonal operator in the continuum limit.
\end{abstract}

\maketitle

\section{Introduction} \label{sec:intro}
Low energy electronic excitations in graphene are described by a massless Dirac operator acting on four--component spinors \cite{Castro2009,DiViMe1984,FeWein2012}. 
The four components take into account a degree of freedom for each of the points in the unit celll of the honeycomb lattice, sometimes called pseudospin, and a degree of freedom for quasiparticles with momenta near the unequivalent Dirac points at the corners of the hexagonal Brillouin zone, the so-called valleys. In the valley--isotropic representation, the Hamiltonian describing these excitations is a direct sum of two two-dimensional Dirac operators, so we define the differential expression

\begin{align}
H= \mqty( T & 0 \\ 0 & T)  \, , \quad T= -i \bm{\sigma} \cdot \grad  = -i\,\mqty( 0 & \partial_1-i\partial_2 \\  \partial_1+i\partial_2 & 0)  \, .
\label{DiracHamiltonian}
\end{align}
Here, we write $\bm{\sigma}=(\sigma_1,\sigma_2)^{\top}$ where $\sigma_1$ and $\sigma_2$ are the first two  Pauli matrices and we use the usual representation, 
\begin{align*}
\sigma_1=\mqty(0&1 \\ 1&0)\, , \quad \sigma_2=\mqty(0 & -i \\ i &0)\, ,
\quad \sigma_3= \mqty(1 &0 \\ 0 & -1) \, .
\end{align*}

When describing electrons confined to a piece of graphene with boundary, suitable boundary conditions must be imposed. Three of these are commonly used in the physics literature: the so-called zigzag, armchair, and infinite mass boundary conditions. The choice of boundary conditions is relevant both from a physical and a mathematical point of view. From the mathematical point of view, they determine the regularity of spinors in the domain of the Hamiltonian and its spectrum. The spectrum and the related transport properties determine the behaviour of the graphene quantum dot when used, for instance, as a single electron transistor.

For the two--dimensional Dirac operator $T$, the most general boundary conditions have been studied by three of us in collaboration with S{\o}ren Fournais in \cite{BenSoStVan2017,BenSoStVan2017-2}. It turns out that there is a one-parameter family of boundary conditions (equation \eqref{eq:2d_bc} below) interpolating between the \emph{zigzag} and \emph{infinite mass} cases. We refer to \cite{BeMon1987, ArrLTRay2017,ArrLTRay2019,StoVu2019, Benhellal2019} for the definition and results on the infinite mass operator and \cite{Sch1995} for early results on the zigzag boundary condition. Further papers on the mathematics of boundary conditions generalize two-dimensional domains with corners \cite{LeOu2018,CaLo2020, PiVdB2021}.
For a discussion of the physical meaning and realization of boundary conditions, we refer to \cite{BrFe2006, PoScKaYaHiNoGe2008, OrBuZo2013, MaMa2011}, the review \cite{Castro2009} and references therein.

\medskip
In  the first part of this article, we study the most general family of local boundary conditions for the four--component operator given, for instance, in \cite{AkhBee2007,AkhBee2008}. To make the paper self--contained, we give a detailed derivation in Appendix~\ref{appendixA}. The main result in Section~\ref{sec:a_unitary_transformation} is a unitary transformation that reduces each of these cases to a block-diagonal operator. This allows us to extend known results about the domain and spectrum for the two-component blocks to the general case. 

In the second part of this article, we specialize to the case of a terminated honeycomb lattice and study the boundary conditions there. For edges perpendicular to the carbon bonds, a block diagonal operator with \emph{zigzag} boundary conditions arises. On the other hand, for edges parallel to the bonds, \emph{armchair} boundary conditions should be imposed, which are not in block-diagonal form. We study them in details in Section~\ref{sec:armchair}, to check for which type of corners, armchair boundary conditions with constant parameters arise. For graphene quantum dots with these corners, the effective Hamiltonian will be unitary equivalent to two copies of $T$ with \emph{infinite mass} boundary conditions. The choice of unit cell and coordinates in the lattice is important for this derivation, hence we recall the derivation of the effective Dirac operator from the tight-binding Hamiltonian in Appendix~\ref{AppendixB}.


\subsection*{Set up and boundary conditions.}
Throughout this paper, $\Omega$ is a $C^2$ domain. For each point $s$ at the boundary, we define the outward normal ${\bm n}(s)=(n_1(s),n_2(s))^T$ and the tangent vector ${\bm t}(s)=(t_1(s),t_2(s))^T$, chosen such that $({\bf n}, {\bf t})$ is positively oriented. 

We first consider boundary conditions for the two--components operator T. It is convenient to write a local boundary condition in the form $m(s) \phi(s) = \phi(s)$, with some Hermitian matrix $m(s)$. In order to give rise to a self-adjoint operator, we can restrict our attention to matrices that are Hermitian, unitary and traceless, which anticommute with the boundary current, 
$$
\left\{ {\bm \sigma} \cdot {\bm n}(s), m(s) \right\}=0 \, .
$$
Such a boundary matrix takes the form
\begin{equation}\label{eq:2d_bc}
    m_\eta (s) := \cos\eta \, \qty(\bm \sigma \cdot \bm t(s)) + \sin\eta \, \sigma_3 \, .
\end{equation}
Here, we write $\bm{\sigma}=(\sigma_1,\sigma_2)^{\top}$ as the usual first two  Pauli matrices and a parameter $\eta \in [0, 2 \pi)$.
We can then define the operator $T_\eta$ that acts as $T$ on the domain
$$
\Dom(T_\eta) := \{\phi \in C^1(\overline{\Omega}, \C^2)\,| \, m_{\eta} \phi =\phi \text{ at } \partial \Omega \}\, .
$$
The following result follows from \cite{BenSoStVan2017} and \cite{Sch1995}.
\begin{theorem*}
    The operator $T_\eta$ is essentially self-adjoint. If $\cos\eta\neq 0$, then the domain of its closure is included in the first Sobolev space $H^1(\Omega, \C^2)$. In the case $\eta = \pi/2$, (resp. $\eta = 3\pi/2$), the domain of the closure is $H^1_0 \oplus \Dom^{\rm max} (\partial_z)$ (resp. $ \Dom^{\rm max}(\partial_{z*})\oplus H^1_0$).
\end{theorem*}

For the graphene Dirac operator $H$, we define a four--parameter family of boundary matrices. 
In order to write out these boundary conditions in a tractable way, we use the Kronecker product notation for $2\times2$ matrices \cite{Moser1996}
\begin{align*}
A\otimes B= \mqty(A_{11} & A_{12} \\ A_{21} & A_{22}) \otimes B =  \mqty(A_{11}B & A_{12}B \\ A_{21}B & A_{22}B) \, .
\end{align*}
We also write $\sigma_0$ for the $2\times 2$ identity matrix, such that for instance
\begin{align*}
    H= \sigma_0 \otimes T \, .
\end{align*}
Finally, we will use throughout the paper a boldface for vectors $\bm v \in \R^2$ and boldface with an arrow for $\vec{\bm v} \in \R^3$.

For $\Gamma:= (\Lambda, \Theta, \theta_\nu, \phi_\nu)$, we define the vectors $\vec{\bm\nu} = (\cos \phi_\nu\cos\theta_\nu, \sin\phi_\nu \cos\theta_\nu, \sin\theta_\nu)^\top$, $\vec{\bm n}_1=(t_1(s)\cos\Theta,t_2(s)\cos\Theta,-\sin\Theta)^\top$ and $\vec{\bm n}_2=(t_1(s)\sin\Theta,t_2(s)\sin\Theta,\cos\Theta)^\top$. For $H$, the boundary matrix takes the form
\begin{align}
M_{\Gamma}:=  \sin\Lambda \, \qty(\sigma_0 \otimes (\vec{\bm{\sigma}}\cdot\vec{\bm{n}}_1 )) +  \cos\Lambda \, \qty( (\vec{\bm{\sigma}}\cdot\vec{\bm{\nu}} ) \otimes (\vec{\bm{\sigma}}\cdot\vec{\bm{n}}_2 ))  \, .
\label{generalM}
\end{align}
We define the corresponding Dirac operators $H_\Gamma$ acting as $H$ on
$$
\Dom(H_\Gamma) := \{\Psi \in C^1(\overline{\Omega}, \C^4)\, | \, M_\Gamma \Psi =\Psi \text{ at } \partial \Omega \}.
$$
Since $H_\Gamma$ anticommutes with the boundary current $\sigma_0\otimes (\bm \sigma \cdot \bm n(s))$, it is a symmetric operator (see Appendix \ref{appendixA} for details).
Our main result is presented in the following theorem.
\begin{theorem}\label{thm:unitary}
    The operator $H_\Gamma$ is unitarily equivalent to the direct sum
    $$
T_{\eta_+} \oplus T_{\eta_-} \quad \text{ with } \eta_{\pm} = -\Theta \pm \qty(\pi/2-\Lambda).
    $$
    In particular it is essentially self-adjoint, and the domain of its closure is included in $H^1$ whenever $\cos\eta_\pm$ are both nonzero. 
    \end{theorem}

The unitary transformation that diagonalizes $H_\Gamma$ is given explicitly in the next section.
Theorem~\ref{thm:unitary} also allows us to obtain the domain of the closure of $H_\Gamma$, and to estimate its spectral gap by using the corresponding result in \cite{BenSoStVan2017-2}, see Corollary~\ref{cor:gap}.
An important special case are armchair boundary conditions. In Section~\ref{sec:armchair}, we show how different angles in the honeycomb lattice give rise, in a continuum limit, to a block-diagonal Dirac operator.

\section{A Unitary Transformation and Its Consequences}\label{sec:a_unitary_transformation}

\begin{proposition}\label{thm:unitarytransformations}
For $\vec{\bm\nu} = (\cos \phi_\nu \cos\theta_\nu, \sin \phi_\nu \cos\theta_\nu, \sin\theta_\nu)^\top$, define 
$$U_{\vec{\bm\nu}}:=\exp (i\frac{\theta_\nu}{2}\sigma_2) \exp(i\frac{\phi_\nu}{2}\sigma_3) \otimes \sigma_0. $$
Then 
     \begin{align*}
    U_{\vec{\bm\nu}} M_{\Gamma} U_{\vec{\bm\nu}}^* =    m_{\eta_{+}}\oplus m_{\eta_{-}} ,  \quad \text{ and } \quad 
    U_{\vec{\bm\nu}} H_{\Gamma} U_{\vec{\bm\nu}}^* =    T_{\eta_{+}}\oplus T_{\eta_{-}},
    \end{align*}
  with $\eta_\pm:= -\Theta \pm \qty(\pi/2-\Lambda)$.
\end{proposition}
\begin{proof}
We will frequently use the property
$$
 (A \otimes B)(C \otimes D) = (AC)\otimes(BD).
$$
As the first step, we consider the matrix 
\begin{align*}
  U_{\phi_\nu}:= e^{i\frac{\phi_\nu}{2}\sigma_3}  \otimes \sigma_0 \, 
\end{align*}
defining a unitary transformation. 
This transformation can be interpreted as a clockwise rotation of the $(x,y)$ plane by an angle $\phi_\nu$. 
The first term of $M_{\Gamma}$ is invariant under this transformation, while for the second term, we have that
 \begin{align*}
     e^{i\frac{\phi_\nu}{2}\sigma_3} (\vec{\bm{\sigma}}\cdot\vec{\bm{\nu}} ) e^{-i\frac{\phi_\nu}{2}\sigma_3} 
     &= \mqty( e^{i\frac{\phi_\nu}{2}} & 0 \\ 0 & e^{-i\frac{\phi_\nu}{2}}) \mqty(\cos\theta_\nu & \sin\theta_\nu e^{-i\phi_\nu} \\  \sin\theta_\nu e^{i\phi_\nu} & -\cos\theta_\nu)  \mqty( e^{-i\frac{\phi_\nu}{2}} & 0 \\ 0 & e^{i\frac{\phi_\nu}{2}})
     \\
     &= \mqty(\cos\theta_\nu & \sin\theta_\nu \\  \sin\theta_\nu & -\cos\theta_\nu) \, .
 \end{align*}
One could therefore restrict our parameters to the case $\phi_\nu=0$, i.e., confining $\vec{\bm{\nu}}$ to the $(x,z)$ plane. Now, we write
\begin{align*}
  \mqty(\cos\theta_\nu & \sin\theta_\nu \\  \sin\theta_\nu & -\cos\theta_\nu) =  \cos\theta_\nu \, \sigma_3+ \sin\theta_\nu \, \sigma_1 =      e^{-i\frac{\theta_\nu}{2}\sigma_2} \sigma_3 e^{i\frac{\theta_\nu}{2}\sigma_2} \, ,
\end{align*}
which motivates the definition $U_{\theta_\nu}:=e^{i\frac{\theta_\nu}{2}\sigma_2}\otimes\sigma_0$.  This matrix defines a unitary transformation that leaves the first term of $M_{\Gamma}$ invariant and it transforms the second term of $M_{\Gamma}$ into the case $\theta_\nu=0$ (i.e., $\vec{\bm{\nu}}=\hat{\bm{z}}$). 

After the two transformations, we obtain 
\begin{align*}
    U_{\theta_\nu} U_{\phi_\nu} M_{\Gamma} U_{\phi_\nu}^\ast  U_{\theta_\nu}^\ast =  \sin\Lambda \, \qty(\sigma_0 \otimes (\vec{\bm{\sigma}}\cdot\vec{\bm{n}}_1 )) +  \cos\Lambda \, \qty( \sigma_3\otimes (\vec{\bm{\sigma}}\cdot\vec{\bm{n}}_2 ))= m_{\eta_+}\oplus m_{\eta_-} \, .
\end{align*}

Using the parameterization of $\vec{\bm{n}}_1$ and $\vec{\bm{n}}_2$ we get
\begin{align*}
    m_{\eta_\pm} &= \qty(\sin\Lambda\cos\Theta\pm \cos\Lambda\sin\Theta )(\bm{\sigma}\cdot \bm{t}) + \qty(-\sin\Lambda\sin\Theta \pm \cos\Lambda\cos\Theta) \sigma_3  \\
    &=\sin (\Lambda\pm\Theta)(\bm{\sigma}\cdot \bm{t}) \pm \cos(\Lambda \pm \Theta)\sigma_3 \, ,
\end{align*}
so $\eta_\pm= -\Theta \pm \qty(\pi/2-\Lambda)$.
Finally, the differential expression $H= T \oplus T$ is invariant under the transformation $U_{\vec{\bm\nu}}$, which maps $\Dom(H_\Gamma)$ onto $\Dom(T_{\eta_+}) \oplus \Dom(T_{\eta_-})$
\end{proof}

A direct consequence of the unitary equivalence is a description of the domain of the closure of the operator $H_\Gamma$

\begin{corollary}\label{cor:domain}
   
    For $\Gamma=( \Lambda,\Theta, \theta_\nu , \phi_\nu)$, define $\eta_\pm= -\Theta \pm \qty(\pi/2-\Lambda)$ as before. If $\cos\eta_+\neq 0$ and $\cos\eta_-\neq 0$, then the closure $\bar H_\Gamma$ has domain included in the first Sobolev space $H^1(\Omega)$.
    In all cases, the domain of $\bar H_\Gamma$ is given by $U_{\vec{\bm{ \nu}}} \Dom(T_{\eta_+})\oplus \Dom(T_{\eta_-})$.
\end{corollary}

Next, we show that the lowest positive eigenvalue has a lower bound that only depends on the area of the domain and on the parameters $\Lambda$ and $\Theta$ that define the boundary conditions. For that purpose, it will be helpful to define the function
\begin{align}
B_\eta:= \min\qty(|\cos\eta/(1-\sin\eta)|, |(1-\sin\eta)/\cos\eta|) \label{B} \, ,
\end{align}
for $\eta \in (0,2\pi)\setminus\{\pi/2, 3\pi/2\}$. 

\begin{corollary}\label{cor:gap}
   
For $\Gamma=(\Theta, \Lambda, \phi_\nu, \theta_\nu )$, define  $\eta_\pm= -\Theta \pm \qty(\pi/2-\Lambda)$ as before. If $\cos\eta_+\neq 0$ and $\cos\eta_{-}\neq 0$, then any eigenvalue $\lambda$ of $H_{{\Gamma}}$ satisfies
\begin{align*}
\lambda^{2} \geq \dfrac{2\pi}{|\Omega|} \min\qty{B_{\eta_+}^{2}, B_{\eta_-}^{2} } \, . 
\end{align*}
\end{corollary}

\begin{proof}
If $\cos(\eta)\neq 0$, the  bound
\begin{align*}
\norm{T_\eta \phi}^{2} \geq \dfrac{2\pi}{|\Omega|} B_{\eta}^{2} \norm{\phi}^{2} \, 
\end{align*}
holds for all $\phi\in\mathrm{Dom}(T_\eta)$, see \cite[Theorem 1]{BenSoStVan2017-2}, where the method from \cite{Bar1992} is applied in the Euclidean case with boundary. Using this inequality and the unitary equivalence obtained in Theorem \ref{thm:unitary},  we obtain that 
\begin{align*}
\norm{H_{\Gamma}\Psi}^2=\norm{(UHU^{\ast})U\Psi}^2 =\norm{T_{\eta_+}\widetilde{\phi}_1}^2 + \norm{T_{\eta_+}\widetilde{\phi}_2}^2 &\geq \dfrac{2\pi}{|\Omega|}\qty(B_{\eta_{+}}^2 \norm{\widetilde{\phi}_1}^2+ B_{\eta_{-}}^2\norm{\widetilde{\phi}_2}^2) \, 
\end{align*}
for all $\Psi\in\mathrm{Dom}(H_\Gamma)$, where $\widetilde{\Psi}=U\Psi =\qty(\widetilde{\phi}_1 , \widetilde{\phi}_2)^{\top}\in\mathrm{Dom}(T_{\eta_+}\oplus T_{\eta_-})$. We complete the proof by taking the minimum of both functions in the last inequality.
\end{proof}


\section{Boundary Conditions for Armchair Edges}\label{sec:armchair}
In this section, we study boundary conditions arising from the tight-binding model for a terminated honeycomb lattice. 
Our goal is to obtain the boundary condition that holds in the discrete setting and express it in the parametric form $M_\Gamma$.
Then, in a formal scaling limit, the tight--binding operator on the domain under consideration converges to a Dirac operator with this boundary condition.

In Appendix~\ref{AppendixB} we recall the derivation of the Dirac operator from the tight--binding model and in Figure~\ref{fig:lattice_conventions}, we show our conventions for the lattice vectors and unit cell.
To obtain the effective Dirac operator in form \eqref{DiracHamiltonian}, we are led to define the $4$--spinor 
\begin{align*}
\Psi := (\Psi_A^{+}, - i \Psi_B^{+}, i \Psi_B^{-}, - \Psi_A^{-})^{\top},  
\end{align*}
where $A,B$ index sublattices and $\pm $ the Dirac points.
A boundary condition in the tight-binding model arises from the requirement that the wavefunction vanishes at the edge sites. For simplicity, here and in the following we write \emph{edge sites} to refer to the lattice sites just outside the edge (the red sites in Figure~\ref{fig:lattice_conventions}).
When a polygon or sector has \emph{zigzag} edges with $A$-sites on the edge, the boundary condition reads simply
\begin{align*}
 \Psi = (0, - i \Psi_B^{+}, i \Psi_B^{-}, 0)^{\top}   
\end{align*}
and we obtain $M_{\mathrm zigzag}= - \sigma_3 \otimes \sigma_3$. For $B$-sites at the outside, the sign flips.

\begin{figure}
    \centering
     \hspace*{\fill}
    \includegraphics{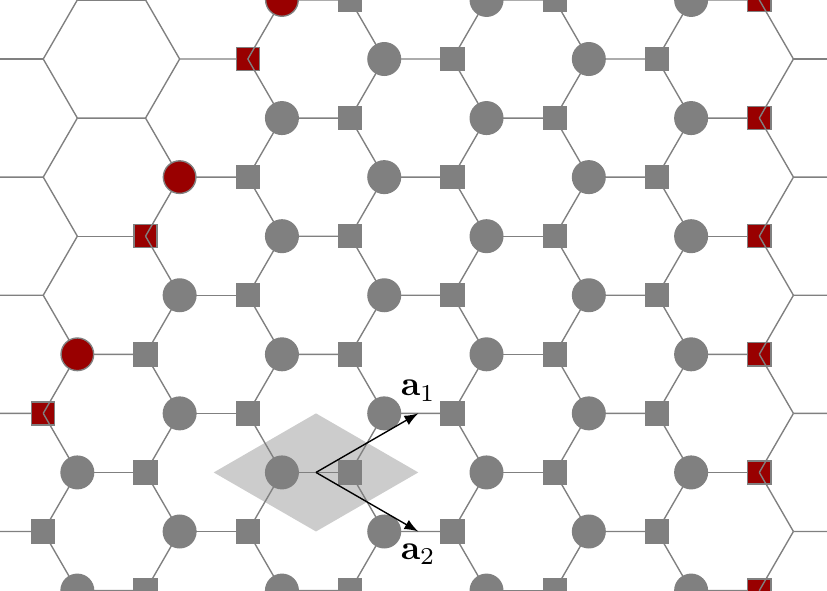}
     \hspace*{\fill}
    \includegraphics{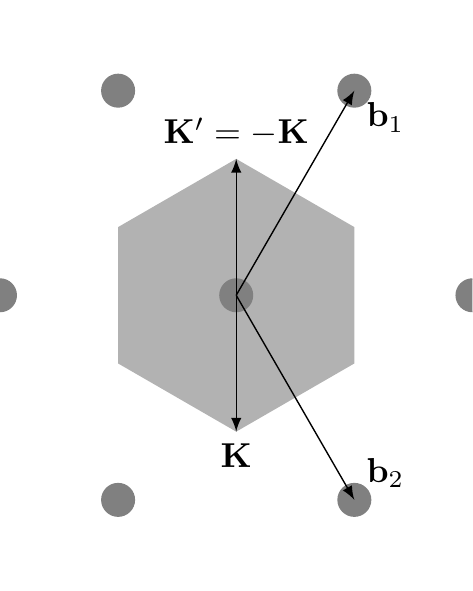}
     \hspace*{\fill}
    \caption{Conventions for the honeycomb lattice and its dual lattice. $A$-sites are circles, $B$-sites squares. At the zigzag edge at the right, the wavefunction vanishes on the red $B$-sites. Along the armchair edge at the left, the wavefunction vanishes on red sites that belong to either sublattice. In this case, the $A$ and $B$ sites along the edge are in different unit cells.}
    \label{fig:lattice_conventions}
\end{figure}

For armchair boundary conditions, the situation is somewhat more involved.
if $\bm{r}_{nm} := n \bm{a}_1 + m \bm{a}_2$ is the position of the corresponding to the $A$ or $B$ site at the edge,
then we need that the sum of contributions from both valleys cancels,
\begin{align*}
\Psi_X^{+}(\bm{r}_{nm})= - e^{i(\bm{K}'-\bm{K})\cdot \bm{r}_{nm}}\Psi_X^{-}(\bm{r}_{nm}), \quad X= A,B.
\end{align*}
We use that $\bm{K}- \bm{K}'=2 (\bm{b}_2 -\bm{b}_1)/3$, where the reciprocal lattice vectors are defined in Figure~\ref{fig:lattice_conventions}.
Inserting this, the boundary condition for the components at the edge is 
\begin{align}\label{eq:armchair_tb}
   \Psi_{X}^{+}(\bm{r}_{nm})= - e^{i \frac{4\pi}{3} (m-n)}\Psi_{X}^{-}(\bm{r}_{nm}), \quad X= A,B. 
\end{align}
In order for this boundary condition to be meaningful in a scaling limit, we need that $e^{i \frac{4\pi}{3} (n-m)}$ is constant when $\bm{r}_{nm}$ varies over the sites of the edge under consideration. This  means that $n-m$ is constant modulo $3$, and this precisely selects the armchair edges, whose equations in terms of the integers $n,m$ are given in Figure~\ref{fig:wedge_example}.

For each armchair edge, the prefactors $e^{i \frac{4\pi}{3} (m-n)}$ take different values on $A$ and $B$ sublattices, that depend on the intercept of the edge. 
All these boundary conditions are unitary equivalent to a block-diagonal one in view of our previous theorem.
However, these precise values become relevant when studying domains bounded by several armchair edges. The question is then whether a unitary transformation that simultaneously diagonalizes the boundary condition for each edge exists. In order to find such a transformation, we have to put the boundary condition on $\Psi_X^{+}, \Psi_X^{-}$ given by \eqref{eq:armchair_tb} into the form $M\Psi = \Psi$ for a matrix $M_\Gamma$ as defined in \eqref{generalM}.

Generally speaking,
an armchair boundary condition takes the form
\begin{align*}
      \Psi_A^{+}(\bm{r}_{nm})&= -\delta_A \Psi_A^{-}(\bm{r}_{nm}) \, ,\\
      \Psi_B^{+}(\bm{r}_{nm})&= -\delta_B \Psi_B^{-}(\bm{r}_{nm}).
\end{align*}
Or in terms of $\Psi$,
$$
M \Psi = \Psi, \quad M:= 
\begin{pmatrix} 
0&0&0&\delta_A^* \\
0&0&\delta_B^*&0\\
0&\delta_B&0&0\\
\delta_A &0&0&0
\end{pmatrix}
$$
with unitary coefficients $\delta_A=  - e^{i \frac{4\pi}{3} (m-n)} $ with $n,m$ the coordinates of an $A$-site at the edge, and analogously for $\delta_B$.
We now check that this matrix $M$ is indeed of the general form presented in \eqref{generalM}. The only possibility for an anti--diagonal matrix is to take $\cos\Lambda= 1$, $\vec{\bm{\nu}}= (\nu_1, \nu_2, 0)$, $\vec{\bm{n}}_2= (t_1, t_2,0)$.
In this case, it is convenient to define complex numbers of unit modulus, $\nu = \nu_1 + i \nu_2$ and similar for $t$, such that
$$
M:= 
\begin{pmatrix} 
0&0&0&\nu^* t^* \\
0&0&\nu^* t&0\\
0&\nu t^*&0&0\\
\nu t &0&0&0
\end{pmatrix}.
$$
We see that both forms are compatible if $\delta_A/\delta_B = t^2$, and that in this case, $\nu = t^*\delta_A$.
The following table shows that this actually happens along each armchair edge.

Now we can study infinite wedges bounded by armchair edges. 
Our problem is to determine the shape of a corner between edges $e_1$ and $e_2$, that gives rise to the same value of $\nu$. 
As illustrated by Figure~\ref{fig:wedge_example}, this happens if and only if $\delta_A(e_1) /\delta_A(e_2) =  t(e_1)/t(e_2)$.
If both edges intersect at an $A$-site, this is not possible, and by symmetry, the same holds for lines intersecting at a $B$-site. It is also possible for the edges to intersect at the centre of a hexagon and a short computation shows that in this case, $\nu$ is indeed constant.
Figure~\ref{fig:wedges} shows the shape of such terminated honeycomb wedges. 

For any \emph{armchair polygon} with these vertices, the boundary condition can be diagonalized. In a scaling limit, the tight-binding Hamiltonian on such a polygon approaches a Dirac operator that is unitary equivalent to two copies of the infinite mass operator. In particular, its spectrum is doubly degenerate and symmetric around zero.

\begin{figure}[]
\begin{minipage}{0.5\linewidth}
\begin{center}
    \begin{tabular}{c| c c c}
 edge direction & equation & $\delta_A/\delta_B$ & $t_1+ i t_2$ \\
\hline
   horizontal &$n-m=c$   & $1$ &  $\pm 1$\\
   60°&$2n + m=c$   & $ e^{-i \frac{4\pi}{3}}$ &  $\pm e^{i\frac{\pi}{3}}$\\
   120° &$2m + n =c$ &$e^{i \frac{4\pi}{3}}$& $\pm e^{i\frac{2\pi}{3}}$
\end{tabular}
\end{center}
\end{minipage}
\begin{minipage}{0.4\linewidth}
    \includegraphics[width=\linewidth]{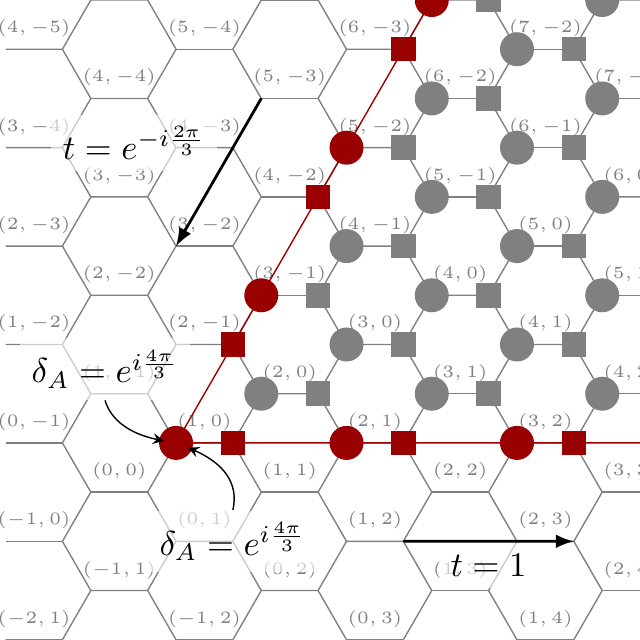}
\end{minipage}

    \caption{Example of a wedge with different values of $\nu$ at each edge. The values of $(n,m)$ for each unit cell are displayed, which allows to compute the equation and values of $\delta_A/\delta_B$ for arbitrary armchair edges. }
    \label{fig:wedge_example}
\end{figure}

\begin{figure}[]

     \centering
    \hspace*{\fill}
    \includegraphics[scale=0.8]{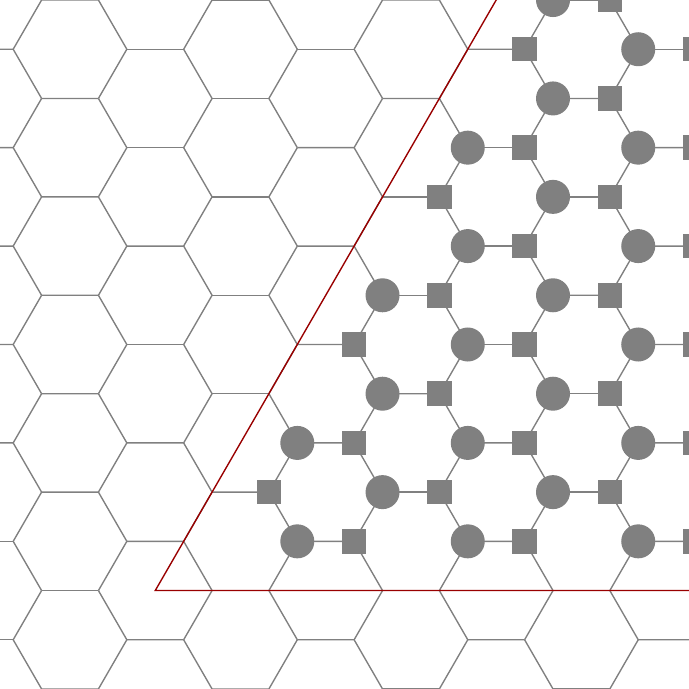}
    \hfill
     \includegraphics[scale=0.8]{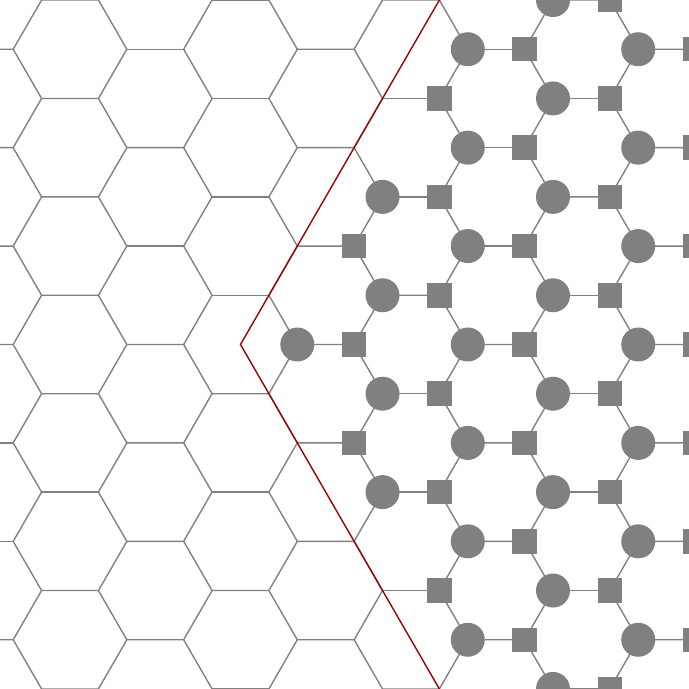}
    \hspace*{\fill}

     \hspace*{\fill}
      \includegraphics[scale=0.8]{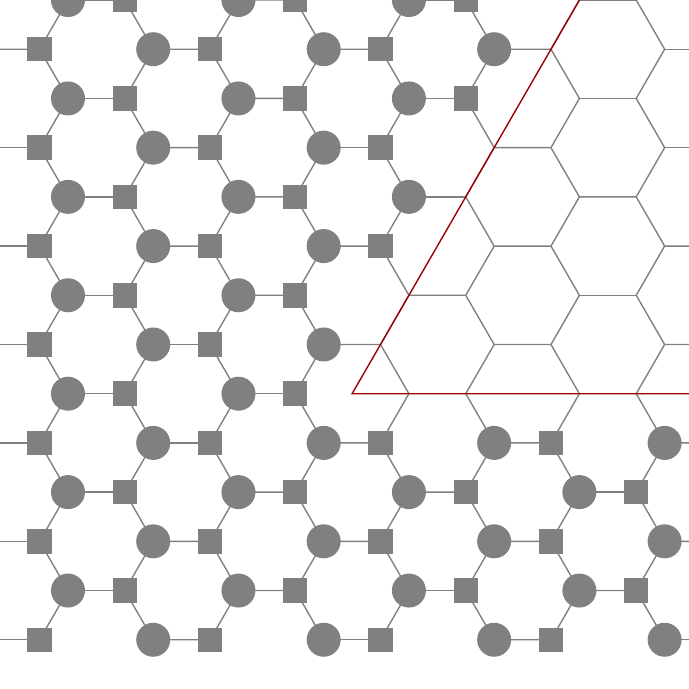}
      \hfill
      \includegraphics[scale=0.8]{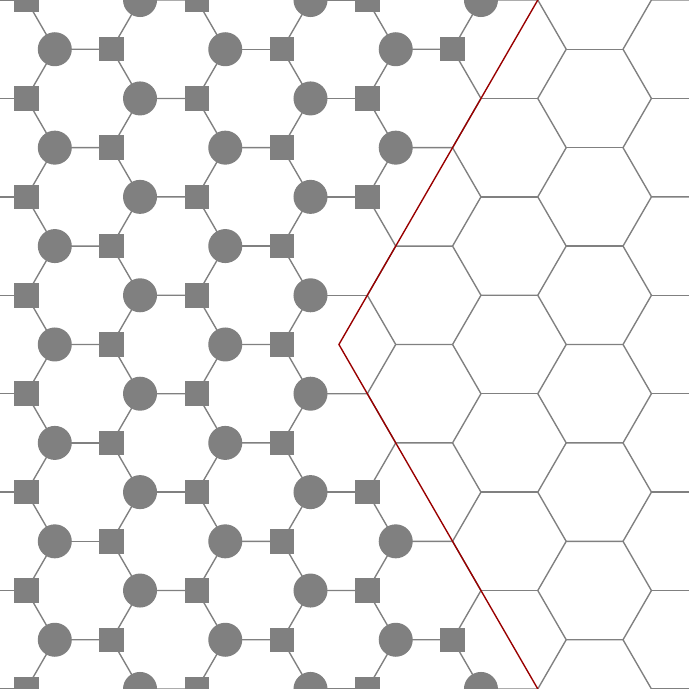}
      \hspace*{\fill}

    \caption{ Armchair wedges that have identical parameters $\nu$ on both edges.}
    \label{fig:wedges}
\end{figure}

\newpage

\setcounter{equation}{0}
\renewcommand{\theequation}{A.\arabic{equation}}

\appendix
\section{Construction of the Boundary Matrices} \label{appendixA}

In this appendix, we explicitly derive an expression for the most general {\it admissible} matrix that turns $H$ into a symmetric operator, the family of matrices $M_{\Gamma}$ in equation \eqref{generalM} (cf. \cite{AkhBee2007,AkhBee2008}). By  {\it admissible} we mean a matrix $M_\Gamma$ which is unitary, traceless and self--adjoint. Furthermore, using Green's identity for $u,v  \in  C_{ }^{1}(\overline\Omega,\mathbb{C}^4)$ we have
 \begin{align}
 \langle u, Hv\rangle & = - i \int_\Omega \qty(u, \sigma_0 \otimes (\bm{\sigma} \cdot \grad) v )_{\mathbb{C}_{ }^{4}} \nonumber \, ,\\
 &= -i \int_\Omega \grad{} \cdot \qty(u, \sigma_0 \otimes \bm{\sigma}  v )_{\mathbb{C}_{ }^{4}} +i \int_\Omega \qty(\sigma_0 \otimes (\bm{\sigma} \cdot \grad) v ,u)_{\mathbb{C}_{ }^{4}} \nonumber  \, , \\
 &= \langle Hu, v\rangle - i \int_{\partial \Omega}  \qty(u, \sigma_0 \otimes (\bm{\sigma}\cdot \bm{n})  v )_{\mathbb{C}_{ }^{4}} \label{Hsymmetry} \, .
 \end{align}
Thus, the boundary term in the last expression vanishes if $M_\Gamma$ anticommutes with the normal current to the boundary $\sigma_0\otimes (\bm{\sigma}\cdot \bm{n})$, i.e., 
 \begin{align*}
     \qty(u, \sigma_0 \otimes (\bm{\sigma}\cdot \bm{n})  v )_{\mathbb{C}_{ }^{4}} = \dfrac{1}{2} \qty ( \qty(M_{\Gamma} u, \sigma_0 \otimes (\bm{\sigma}\cdot \bm{n})  v )_{\mathbb{C}_{ }^{4}} + \qty(u, \sigma_0 \otimes (\bm{\sigma}\cdot \bm{n})  M_{\Gamma} v )_{\mathbb{C}_{ }^{4}} ) \, .
 \end{align*}
Summing up, we look for a matrix $M_\Gamma$ satisfying
\begin{subequations}
	\begin{align}
	M_{\Gamma}^\ast = M_{\Gamma} \, \text{,} \, M_{\Gamma}^2=1  \, & , \text{ and}, \label{eq:Mselfadjoint}\\
	\{M_{\Gamma},\sigma_0 \otimes (\bm{\sigma}\cdot \bm{n})\}&=0 \label{ap:anticommutation} \, .
	\end{align}
\end{subequations}
First, we can express a Hermitian  $4\times4$ matrix as a linear combination of the Kronecker product between the $2\times 2$ Pauli matrices,
\begin{align*}
M_{\Gamma}= \sum_{i,j=0}^{3} c_{ij} ( \sigma_i\otimes \sigma_j) \, ,
\end{align*}
where $c_{ij}\in\mathbb{R}$ because $M_{\Gamma}$ is self--adjoint. For $\bm{\vec{a}},\bm{\vec{b}}\in \mathbb{R}^3$, the following properties of  the Pauli matrices are useful to establish the conditions on these real coefficients $c_{ij}$,
\begin{subequations}
	\begin{align}
	(\bm{\vec{\sigma}}\cdot\bm{\vec{a}})(\bm{\vec{\sigma}}\cdot\bm{\vec{b}})&= (\bm{\vec{a}} \cdot \bm{\vec{b}})\,\sigma_0 +i\bm{\vec{\sigma}}\cdot (\bm{\vec{a}}\times \bm{\vec{b}})    \, , \label{ap:relation1}\\
	\{ \sigma_j,\bm{\vec{\sigma}}\cdot\bm{\vec{a}} \} &= 2 a_j \sigma_0  (1-\delta_{j0})+ 2 (\bm{\vec{\sigma}}\cdot\bm{\vec{a}}) \delta_{j0}  \label{ap:relation2} \, .	
	\end{align}
\end{subequations}
Using the anticommutation relations \eqref{ap:anticommutation} and \eqref{ap:relation2}, we obtain
\begin{align*}
\{M_{\Gamma},\sigma_0 \otimes (\bm{\sigma}\cdot \bm{n})\} 
&= 2\sum_{i=0}^{3} \sigma_i \otimes \qty( c_{i0} (\bm{\sigma}\cdot \bm{n})  + \sigma_0 \sum_{j=1}^{3} c_{ij} n_{j}   ) = 0 \, .
\end{align*}
Thus, the term in parenthesis must vanish. With the definition $\bm{\vec{n}}=(n_1(s),n_2(s),0)^{\top}$, we obtain that $c_{i0}=\bm{\vec{c}}_i \cdot \bm{\vec{n}}=0$ for all $i\in\{0,1,2,3\}$, with $\bm{\vec{c}_i}=(c_{i1},c_{i2},c_{i3})^{\top}$. Hence,
\begin{align*}
M_{\Gamma}= \sum_{i=0}^{3}  ( \sigma_i \otimes \bm{\vec{\sigma}}\cdot \bm{\vec{c}}_i )
= \mqty( \bm{\vec{\sigma}}\cdot(\bm{\vec{c}}_0 + \bm{\vec{c}}_3)  & \bm{\vec{\sigma}}\cdot(\bm{\vec{c}}_1 - i\bm{\vec{c}}_2)  \\
\bm{\vec{\sigma}}\cdot(\bm{\vec{c}}_1 + i\bm{\vec{c}}_2)  & \bm{\vec{\sigma}}\cdot(\bm{\vec{c}}_0 - \bm{\vec{c}}_3)  )  \, .
\end{align*}
Using the relation \eqref{ap:relation1}, we explicitly obtain that
\begin{align*}
M_{\Gamma}^2 &= \mqty( \sigma_0\qty((\bm{\vec{c}}_0 + \bm{\vec{c}}_3)^2   + \bm{\vec{c}}_1^{\ 2} + \bm{\vec{c}}_2^{\ 2} )  +2 \bm{\vec{\sigma}}\cdot (\bm{\vec{c}}_1\times \bm{\vec{c}}_2)& 
2\sigma_0 \bm{\vec{c}}_0 \cdot (\bm{\vec{c}}_1 - i\bm{\vec{c}}_2) + 2i \bm{\vec{\sigma}}\cdot (\bm{\vec{c}}_3 \times (\bm{\vec{c}}_1 - i\bm{\vec{c}}_2))  \\
2\sigma_0\bm{\vec{c}}_0 \cdot (\bm{\vec{c}}_1 + i\bm{\vec{c}}_2) + 2i \bm{\vec{\sigma}}\cdot( \bm{\vec{c}}_3 \times (\bm{\vec{c}}_1 + i\bm{\vec{c}}_2))& 
\sigma_0((\bm{\vec{c}}_0 - \bm{\vec{c}}_3)^2   +   \bm{\vec{c}}_1^{\ 2} + \bm{\vec{c}}_2^{\ 2} )  +2 \bm{\vec{\sigma}}\cdot(\bm{\vec{c}}_1\times \bm{\vec{c}}_2)
) \, .
\end{align*}
The condition $M_{\Gamma}^2=1$ implies that $\bm{\vec{c}}_0=c_0\bm{\vec{n}}_1$ is orthogonal to $\bm{\vec{c}}_1,\bm{\vec{c}}_2,\bm{\vec{c}}_3$, $\bm{\vec{c}}_i =c_i \bm{\vec{n}}_2$ ($i\in\{1,2,3\}$) for some unit vector $\bm{\vec{n}}_2$ orthogonal to $\bm{\vec{n}}_1$, and  $c_0^2+c_1^2+c_2^2+c_3^2=1$.
It follows that 
\begin{align*}
M_{\Gamma}=  \sin\Lambda \, \qty(\sigma_0 \otimes (\bm{\vec{\sigma}}\cdot\bm{\vec{n}}_1 )) + \cos\Lambda \, \qty( (\bm{\vec{\sigma}}\cdot\bm{\vec{\nu}} ) \otimes  \, (\bm{\vec{\sigma}}\cdot\bm{\vec{n}}_2 ))   \, ,
\end{align*}
where $\bm{\vec{\nu}},\bm{\vec{n}}_1, \bm{\vec{n}}_2$ are three--dimensional unit vectors such that $\bm{\vec{n}}_1 \cdot\bm{\vec{n}}_2 =\bm{\vec{n}}_1 \cdot \bm{\vec{n}}=\bm{\vec{n}}_2 \cdot \bm{\vec{n}}=0$ and $\Lambda \in \R$. We paramertrize $\vec{\bm{\nu}}=(\cos\phi_\nu\cos\theta_\nu,\sin\phi_\nu\cos\theta_\nu,\sin\theta_\nu)$.

\setcounter{equation}{0}
\renewcommand{\theequation}{B.\arabic{equation}}

\section{Derivation of the Dirac Equation }\label{AppendixB}
We use the conventions introduced in Figure~\ref{fig:lattice_conventions}. Integer indices $n,m$  label each unit cell, the position of its centre is defined as $\bm{r}_{nm}:= n \bm{a}_1 + m \bm{a}_2$. 
In a scaling limit, $\bm{r}_{nm}$ becomes a continuous variable and therefore it is convenient to write the discrete wavefunction at a lattice site as $\psi_A(\bm{r})$ and $\psi_B(\bm{r})$. The tight--binding Hamiltonian at a site $A(B)$ depends on the sum of the wave--function at its nearest neighbours on the $B(A)$--sublattice. 
\begin{align*}
    \left(H_{\mathrm{t.b.}} \psi\right)_A(\bm{r})
    &= t\qty(\psi_B(\bm{r}) + \psi_B(\bm{r}-\bm{a}_1)+\psi_B(\bm{r}-\bm{a}_2)) \, , \\
    \left(H_{\mathrm{t.b.}} \psi\right)_B(\bm{r})
    &= t\qty(\psi_A(\bm{r}) + \psi_A(\bm{r}-\bm{a}_1)+\psi_A(\bm{r}-\bm{a}_2)) \, .
\end{align*}
The energies are given by $\pm \abs{f(\bm{k})}$, with $f(\bm{k})=t\qty(1+e^{-i\bm{k}\cdot\bm{a}_2}+e^{-i\bm{k}\cdot\bm{a}_1})$ and $\bm{k}$ in the first Brillouin zone (FBZ). 

The restriction to low energies amounts to replacing each of these wave--functions by plane waves with momenta $\pm\bm{K}$, which are the so--called Dirac points in the FBZ, defined as the wave--vectors where the energy vanishes: $f(\pm\bm{K})=0$. To simplify the calculations, we have chosen the non--equivalent Dirac points as the two corners of the FBZ lying in the vertical axis (see Figure~\ref{fig:lattice_conventions}): $\xi\bm{K}=\xi\qty(0,-\frac{4\pi}{3a})^\top$, where $\xi=\pm$ is the valley index. Thus, the Ansatz for the wavefunction $\psi_X$ becomes
\begin{align*}
\psi_X(\bm{r})=  e^{i\bm{K}\cdot\bm{r}} \Psi_{X}^{+}(\bm{r})+e^{-i\bm{K}\cdot\bm{r}}\Psi_{X}^{-}(\bm{r}), \quad X=A,B. 
\end{align*}
Replacing the above in the tight--binding Hamiltonian at a site $A$, we get
\begin{align*}
    \left(H_{\mathrm{t.b.}} \psi\right)_A(\bm{r})
    &= t\sum\limits_{\xi=\pm} e^{i\xi\bm{K}\cdot\bm{r}} \left(\Psi_{B}^{\xi}(\bm{r}) +  e^{-i\xi\bm{K}\cdot\bm{a_1}}\Psi_{B}^{\xi}(\bm{r}-\bm{a}_1)+e^{-i\xi\bm{K}\cdot\bm{a_2}}\Psi_{B}^{\xi}(\bm{r}-\bm{a}_2)\right) \, .
\end{align*}
Next, we aproximate $\Psi_{B}^{\xi}(\bm{r} - \bm{a}_j)$ by its first--order Taylor expansion. The constant terms vanish by the definition of the Dirac points, and we are left with
\begin{align*}
    \left(H_{\mathrm{t.b.}} \psi\right)_A(\bm{r})
    &\approx -t\sum\limits_{\xi=\pm} e^{i\xi\bm{K}\cdot\bm{r}}\left(e^{-i\xi\bm{K}\cdot\bm{a_1}}\bm{a}_{1}+e^{-i\xi\bm{K}\cdot\bm{a_2}}\bm{a}_{2}\right) \cdot\grad \Psi_{B}^{\xi}(\bm{r}) \, , \\
    &= -\dfrac{\sqrt{3}}{2}ta \sum\limits_{\xi=\pm} e^{i\xi\bm{K}\cdot\bm{r}}\qty( \partial_1-i\xi\partial_2 ) \Psi_{B}^{\xi}(\bm{r}) \, .
\end{align*}
In the last line we used that $\bm{a}_1=a/2(\sqrt{3},1)^\top$ and $\bm{a}_2=a/2(\sqrt{3},-1)^\top$. This expression leads to the definition $v_F=\sqrt{3}\,ta/2$, the Fermi velocity in graphene. By symmetry of the operator (or the analogous computation), for a $B$-site we obtain
\begin{align*}
    \left(H_{\mathrm{t.b.}} \psi\right)_B(\bm{r})
    &\approx v_F \sum\limits_{\xi=\pm} e^{i\xi\bm{K}\cdot\bm{r}}\qty( \partial_1+i\xi\partial_2 ) \Psi_{A}^{\xi}(\bm{r}) \, .
\end{align*}
Thus, upon defining the spinor  $\Psi(\vb{r}) =\qty(\Psi^{+}_A(\bm{r}),-i\Psi^{+}_B(\bm{r}),i\Psi^{-}_B(\bm{r}),-\Psi^{-}_A(\bm{r}))^\top$ in terms of the four amplitudes, we obtain the effective Hamiltonian that acts as
\begin{align*}
    H\Psi(\vb{r})= -iv_F \qty(\sigma_0 \otimes (\bm{\sigma}\cdot\grad)) \Psi(\vb{r}) \, .
\end{align*}
Finally, we set $v_F=1$ to recover equation \eqref{DiracHamiltonian}. 

\section*{Acknowledgments}
\thanks{The work of R.B. has been supported by Fondecyt (Chile) Project \# 120--1055.  The work of E.S has been partially funded by Fondecyt (Chile) Project \# 114--1008. The work of C.V. has been supported by Becas Chile and Fondecyt Projects \# 116--0856 and \# 120--1055. The work of H. VDB. has been partially supported by  Fondecyt Project \# 1122–0194 and by the Centre for Mathematical Modeling, ANID Basal grant \# FB210005.
}


\begin{thebibliography}{10}
	
	

\bibitem{AkhBee2007}
Akhmerov, A. R., Beenakker,  C. W. J.: Detection of Valley Polarization in Graphene by a Superconducting Contact. Phys. Rev. Lett. {\bf 98}, 157003 (2007).

\bibitem{AkhBee2008}
Akhmerov, A. R., Beenakker, C. W. J.: Boundary conditions for Dirac fermions on a terminated honeycomb lattice. Phys. Rev. B {\bf 77}, 085423 (2008).

\bibitem{ArrLTRay2017}
Arrizabalaga, N., Le Treust, L., Raymond, N.: On the MIT bag model in the non-relativistic limit. Commun. Math. Phys. {\bf 354}, 641--669 (2017). 


\bibitem{ArrLTRay2019}
Arrizabalaga, N., Le Treust, L., Mas, A., Raymond, N.: The MIT bag model as an infinite mass limit. Journal de l'Ecole Polytechnique -- Math\'ematiques, Tome {\bf 6}, 329--365 (2019). 

\bibitem{Bar1992} Bär, C.
Lower eigenvalue estimates for Dirac operators , Math. Ann.
{\bf 293} no. 1, 39--46 (1992).

\bibitem{BaCoLeSt2019}
Barbaroux, JM., Cornean, H., Le Treust, L., Stockmeyer, E.: Resolvent Convergence to Dirac Operators on Planar Domains. Ann. Henri Poincar\'e {\bf 20}, 1877--1891 (2019).

\bibitem{Bena2009}
Bena, C. and Montambaux, G.: Remarks on the tight-binding model of graphene. New Journal of Physics, {\bf 11}(9), p.095003 (2009).

\bibitem{BenSoStVan2017}
Benguria, R. D., Fournais, S., Stockmeyer, E., Van Den Bosch, H.: Self--Adjointness of two--dimensional Dirac Operators on Domains. Ann. Henri Poincar\'e {\bf 18}, 1371--1383 (2017).


\bibitem{BenSoStVan2017-2}
Benguria, R. D., Fournais, S., Stockmeyer, E., Van Den Bosch, H.: Spectral Gaps of Dirac Operators Describing Graphene Quantum Dots. Math. Phys. Anal. Geom. {\bf 20}, 11 (2017).


\bibitem{Benhellal2019}
Benhellal, B.: Spectral Asymptotic for the Infinite Mass Dirac Operator in bounded domain (2019). Preprint: \url{arXiv:1909.03769}

\bibitem{BeMon1987}
Berry, M. V., Mondragon, R. J.: Neutrino billiards: time--reversal symmetry--breaking without magnetic fields. Proc. R. Soc. London A {\bf 412}, 53--74 (1987). 


\bibitem{BrFe2006}
Brey, L., Fertig,  H. A.: Electronic states of graphene nanoribbons studied with the Dirac equation, Phys. Rev. B {\bf 73}, 235411 (2006).

\bibitem{CaLo2020}
Cassano, B., Lotoreichik, V.: Self--adjoint extensions of the two--valley Dirac operator with discontinuous infinite mass boundary conditions. To appear in Oper. Matrices (2020).

\bibitem{Castro2009}
Castro Neto, A. H., Guinea, F.,  Peres, N. M. R., Novoselov, K. S.,  Geim, A. K.: The electronic properties of graphene, Rev. Mod. Phys. {\bf 81}, 109--162 (2009).

\bibitem{DiViMe1984}
DiVincenzo, D. P., Mele, E. J.: Self-consistent effective-mass theory for intralayer screening in graphite intercalation compounds. Phy. Rev. B, {\bf 29}(4), 1685--1694 (1984). 

\bibitem{FeWein2012}
Fefferman, C. L., Weinstein, M.: Honeycomb lattice potentials and Dirac points. J. Amer. Math. Soc. 25, 1169--1220 (2012).

\bibitem{FrSi2014}
Freitas, P., Siegl, P.: Spectra of graphene nanoribbons with armchair and zigzag boundary conditions. Rev. Math. Phys. {\bf 26}(10), 1450018 (2014).


\bibitem{LeOu2018}
Le Treust, L.,   Ourmi\`eres-Bonafos, T.: Self--adjointness of Dirac operators with infinite mass boundary conditions in sectors. Annales Henri Poincar\'e, {\bf 19}(5): 1465--1487 (2018).

\bibitem{LoOu2019}
Lotoreichik, V., Ourmi\`eres-Bonafos, T.: A sharp upper bound on the spectral gap for graphene quantum dots. Math. Phys. Anal. Geom. {\bf 22}, 13 (2019). 



\bibitem{MaMa2011}
Marconcini, P., Macucci, M. The $k\cdot p$ method and its application to graphene, carbon nanotubes and graphene nanoribbons: the Dirac equation. Riv. Nuovo Cim. 34, 489–584 (2011). 

\bibitem{McFal2004}
McCann, E., Fal'ko, V. I.: Symmetry of boundary conditions of the Dirac equation for electrons in carbon nanotubes. J. Phys. Condens. Matter {\bf 16}(13), 2371--2379 (2004). 

\bibitem{Moser1996}
Moser, B. K.: Linear Models: A Mean Model Approach (Probability and Mathematical Statistics). Springer, New York (1996).


\bibitem{OrBuZo2013}
Orlof, A., Ruseckas, J., Zozoulenko, I.V.: Effect of zigzag and armchair edges on the electronic transport in single--layer and bilayer graphene nanoribbons with defects, Phys. Rev. B {\bf 88}, 125409 (2013).


\bibitem{PiVdB2021}
Pizzichillo, F., Van Den Bosch, H.: Self--adjointness of two--dimensional Dirac operators on corner domains. J. Spectr. Theory 11, no. 3, 1043–-1079 (2021).

\bibitem{PoScKaYaHiNoGe2008}
Ponomarenko, L. A., Schedin, F., Katsnelson, M. I., Yang, R., Hill, E. W., Novoselov, K. S., Geim, A. K.: Chaotic Dirac billiard in graphene quantum dots. Science {\bf 320}, 356--358 (2008). 

\bibitem{Sch1995}
Schmidt, K. M.: A remark on boundary value problems for the Dirac operator. Q. J. Math. Oxf. Ser. (2) {\bf 46}, 509--516 (1995). 


\bibitem{StoVu2019}
Stockmeyer, E., Vugalter, S.: Infinite mass boundary conditions for Dirac operators. Journal of Spectral Theory {\bf 9}(2),  569--600
(2019).





\bibitem{Zak1972}
Zak, J.: The kq-representation in the dynamics of electrons in solids. Solid State Physics {\bf 27}, 1--62 (1972). 


\bibitem{Zhen2007}
Zheng, H., Wang, Z.F., Luo, T., Shi, Q. W., Chen, J.: Analytical study of electronic structure in armchair graphene nanoribbons. Phys. Rev. B {\bf 75}, 165414 (2007).

\end{thebibliography}
\end{document}